\documentclass[english,a4paper]{article}
\usepackage{dcolumn}
\usepackage{bm}
\usepackage{graphicx}
\usepackage{amsmath}
\usepackage{amsfonts}
\usepackage{amssymb}
\usepackage{amsthm}
\usepackage{amstext}
\usepackage{amsbsy}
\usepackage{amsopn}
\usepackage{amscd}
\usepackage{amsxtra}

\newtheorem{theorem}{Theorem}
\newtheorem{lemma}{Lemma}

\newtheorem{proposition}{Proposition}

\newcommand{\red}[1]{{\color{red} #1}}

\newtheorem{definition}{Definition}

\def\openone{\leavevmode\hbox{\small1\kern-3.3pt\normalsize1}}

\usepackage[usenames,dvipsnames]{color}

\begin{document}
\title{Complete positivity, positivity and long-time asymptotic behavior in a two-level open quantum system}
\author{G. Th\'eret\footnote{Laboratoire Interdisciplinaire Carnot de Bourgogne (ICB), UMR 6303 CNRS-Universit\'e Bourgogne-Franche Comt\'e, 9 Av. A.
Savary, BP 47 870, F-21078 Dijon Cedex, France}, D. Sugny\footnote{Laboratoire Interdisciplinaire Carnot de
Bourgogne (ICB), UMR 6303 CNRS-Universit\'e Bourgogne-Franche Comt\'e, 9 Av. A.
Savary, BP 47 870, F-21078 Dijon Cedex, France, dominique.sugny@u-bourgogne.fr}}

\maketitle

\begin{abstract}
We study the concepts of complete positivity and positivity in a two-level open quantum system whose dynamics are governed by a time-local quantum master equation. We establish necessary and sufficient conditions on the time-dependent relaxation rates to ensure complete positivity and positivity of the dynamical map. We discuss their relations with the non-Markovian behavior of the open system. We also analyze the long-time asymptotic behavior of the dynamics as a function of the rates. We show under which conditions on the rates the system tends to an equilibrium state. Different examples illustrate this general study.
\end{abstract}

%\maketitle

%\date{\today}

\section{Introduction}
A perfect isolation of a quantum system from its environment is not possible in realistic physical
processes. The interaction with the environment is generally detrimental and leads to a loss of
information and quantum correlations~\cite{breuerbook,weissbook,alickibook,vega2017}. In some cases, this property can be captured by modeling empirically the
dynamics of the system by a time-local master equation in Redfield form. This differential equation is characterized by different relaxation rates and frequency transitions which may be time-dependent. In this paper, we are not interested in the derivation of such functions and we assume, on the basis of experimental data and knowledge of the dynamical systems, that these functions exist and are sufficiently smooth~\cite{relaxrates2,wonderen1999,laine2012,megier2020}. The dynamical evolution of this open system must satisfy specific properties such as complete positivity (CP) and positivity (P) to ensure that the state remains physically valid at any time. Complete positivity is due to the assumption of factorized initial conditions between the system and the bath~\cite{carteret2008,pechukas1994,alicki1995}, while positivity is required for defining the density operator of the reduced system. CP is also necessary to preserve positivity of the dynamics when the system is entangled with other quantum degrees of freedom. For constant relaxation rates, the impact of these dynamical constraints has been studied extensively both for two-level and larger quantum systems. The relaxation parameters must fulfill different inequalities which can be established~\cite{schirmer2004} by putting the dynamical system in the Gorrini-Kossakovski-Lindblad-Sudarshan (GKLS) form~\cite{lindblad1,lindblad2}. In this case, the process represents a semigroup and a sufficient and necessary condition for CP is that all the GKLS diagonal decay rates $\gamma_i$ are positive. This property can be extended to time-dependent rates when $\gamma_i(t)\geq 0$ giving a sufficient condition of CP. However, in many situations, this description is too restrictive and memory effects due to the non-Markovian behavior of the dynamics have to be taken into account~\cite{wolfprl,RMPbreuer,reviewrivas}. The non-Markovianity (NM) is characterized by the negativity of, at least, one of the coefficients $\gamma_i$ during a given time interval. Note that a large number of quantitative measures has been proposed recently to detect this property on the basis of specific experimental data~\cite{tracedistance,tracedistance2,tracedistanceSB,relaxrates1,measurevol,localnonlocal,rivasprl,mesure_expnaturephys,measureother1,measureother2,measureother3,wibmann2012,measureother4,measureother5,reviewPR,nonmarkoviansaalfrank}. For NM phenomenological master equations, CP and P may be violated, the dynamical map loosing then its physical meaning. In this case, there is no straightforward way to verify CP and P when the relaxation rates are known. In this paper, we propose to study this general problem in the case of a two-level quantum system. We establish necessary and sufficient conditions on the decay rates to ensure CP and P of the dynamics. The positivity of the Choi matrix~\cite{relaxrates1,choimatrix} is used for CP, while a direct computation is performed for the positivity property. Studies on the CP of time-local quantum master equations have been performed in~\cite{maniscalco2007,hall2008}, but for specific NM dynamics.  Works have also been carried out for other representations of the dynamical equation~\cite{budini2004,vacchini2016,breuer2008,vacchini2013,chruscinski2016} such as, e.g., time nonlocal integro-differential equations. In this paper, we also find families of systems where CP and P are equivalent and we exhibit examples for which the dynamical map is only positive. We discuss the link between CP and NM and we introduce a simple family of dynamical systems, the quasi-Markovian ones which are by construction CP, but can be Markovian or NM. Standard examples illustrate this general study. As a byproduct, we finally investigate the long-time asymptotic behavior of the dynamics with respect to the decay rates. We establish under which conditions on the rates the system tends towards the equilibrium state. These different results are interesting from an experimental point of view because they give physical constraints to respect in the construction of an empirical quantum master equation.

The paper is organized as follows. The model system is presented in Sec.~\ref{sec2} with its dynamics in Redfield and GKLS forms. The CP and P of the dynamical map are studied respectively in Sec.~\ref{sec3} and Sec.~\ref{sec4}. Necessary and sufficient conditions on the relaxation rates are established to ensure that such properties are satisfied. Quasi-Markovian systems are introduced in Sec.~\ref{sec5}. The long-time asymptotic behavior of the dynamical system is described in Sec.~\ref{secnew}. Different examples are investigated in Sec.~\ref{sec6}. Conclusion and prospective views are given in Sec.~\ref{sec7}. Proofs about the positivity and the quasi-Markovian behavior of the system are respectively described in Appendices~\ref{appb} and \ref{appC}. Additional results about the asymptotic behavior of the dynamics are shown in Appendix~\ref{secasympnew}.

\section{The model system}\label{sec2}
We consider a two-level open quantum system whose state is described at time $t$ by a density operator $\rho(t)$, i.e. a positive semidefinite Hermitian operator of unit trace. We denote by $\mathcal{H}$ the two-dimensional Hilbert space of the system, spanned by the canonical basis $\{|1\rangle,|2\rangle\}$, and by $\mathcal{S}(\mathcal{H})$ the set of density operators. The quantum dynamical linear map $\Phi_t$ from $\mathcal{S}(\mathcal{H})$ to $\mathcal{S}(\mathcal{H})$ maps by definition the initial state $\rho(0)$ to the state $\rho(t)$ as $\rho(t)=\Phi_t[\rho(0)]$, with $\Phi_0=I$, the unit map at time 0~\cite{RMPbreuer}. The map $\Phi_t$ needs not only to be positive to ensure that $\rho(t)$ is a well-defined density operator, but also completely positive in the assumption of factorized initial conditions between the system and the bath. We recall that a positive map is a map which transforms positive operators into positive operators. The map $\Phi_t$ also preserves hermiticity and the trace of operators, so that it maps a density operator of $\mathcal{S}(\mathcal{H})$ to another density operator of $\mathcal{S}(\mathcal{H})$. The property of CP is defined in the tensor product space $\mathcal{H}\otimes \mathbb{C}^n$ where $n$ is a non-zero positive integer. We consider the map $\Phi_t\otimes I_n$ which operates on this space, $I_n$ being the identity operator on $\mathbb{C}^n$. $\Phi_t$ is said to be CP if $\Phi_t\otimes I_n$ is positive for all $n$. A positive map $\Phi_t$ corresponds to the case $n=1$. It is then obvious that CP implies P. A characterization of CP in terms of Choi matrix is used in Sec.~\ref{sec3}.

The existence of the inverse of $\Phi_t$, $\Phi_t^{-1}$, at all times $t\geq 0$ allows to write the dynamical evolution of the system as a time-local quantum master equation~\cite{relaxrates1}. This assumption is the starting point of our study. The invertibility of $\Phi_t$ allows to define a dynamical map $\Phi_{t,s}$ from times $s$ to $t$, $t\geq s\geq 0$ as $\Phi_{t,s}=\Phi_t\Phi_s^{-1}$. Note that the map $\Phi_{t,s}$ needs not to be CP or P. The dynamical map is said to be CP-divisible (resp. P-divisible) if $\Phi_{t,s}$ is CP (resp. P) for all $t$ and $s$. It can then be shown that $CP$- divisible maps are also Markovian~\cite{RMPbreuer,measureother5,reviewPR}. In this study, we emphasize that we investigate under which conditions $\Phi_t$ is CP or P, which is different from CP or P-divisibility.

In the case of a two-level quantum system, denoting by $(\rho_{11},\rho_{12},\rho_{21},\rho_{22})$ the matrix elements of $\rho$ in the canonical basis, the time evolution can be written in Redfield form~\cite{schirmer2004} as
\begin{equation}\label{eqdyntwo}
\frac{d}{dt}
\begin{pmatrix}
\rho_{11} \\
\rho_{12} \\
\rho_{21} \\
\rho_{22}
\end{pmatrix}
=
\begin{pmatrix}
-\gamma_{21} & 0 & 0 & \gamma_{12}\\
0 & i\omega-\Gamma & 0 & 0\\
0 & 0 & -i\omega-\Gamma & 0\\
\gamma_{21} & 0 & 0 & -\gamma_{12}
\end{pmatrix}
\begin{pmatrix}
\rho_{11} \\
\rho_{12} \\
\rho_{21} \\
\rho_{22}
\end{pmatrix}
\end{equation}
where $\omega$ is the frequency transition of the two-level system, $\gamma_{12}$ and $\gamma_{21}$ are respectively the relaxation rates of the populations from level 2 to 1 and from 1 to 2. The parameter $\Gamma$ describes the dephasing of the coherences. Units such that $\hbar=1$ are used throughout the paper. By construction, note that the trace of the density operator is equal to 1 at any time $t$ if $\textrm{Tr}[\rho(0)]=1$. The frequency transition and the different decay rates are assumed to be sufficiently smooth functions of time. The differential system~\eqref{eqdyntwo} can be put in the GKLS-like form~\cite{lindblad1,lindblad2,relaxrates1} as
\begin{eqnarray}\label{eqlindblad}
& & \frac{d}{dt}\rho(t) = -i[H,\rho(t)] \\
& & +\sum_{j=1}^3 \gamma_j(t) \left( L_j \rho(t) L_j^\dagger - \frac{1}{2}\left\{ L_j^\dagger L_j,\rho(t)\right\}\right).\nonumber
\end{eqnarray}
with $\gamma_1 = \gamma_{12}$, $\gamma_2 = \gamma_{21}$, $\gamma_3 = \Gamma - \frac{\gamma_{+}}{2}$, $\gamma_+=\gamma_{21}+\gamma_{12}$, $\gamma_-=\gamma_{12}-\gamma_{21}$,
$$
L_1 =
\begin{pmatrix}
	0 	& 	1		\\
    	0  	& 	0 		
\end{pmatrix},~
L_2 =
\begin{pmatrix}
	0 	& 	0		\\
    	1  	& 	0 		
\end{pmatrix},~
L_3 = \frac{1}{\sqrt{2}}
\begin{pmatrix}
	1 	& 	0		\\
    	0  	& 	-1 		
\end{pmatrix},
$$
and
$$
H=
\frac{\omega}{2}
\begin{pmatrix}
	-1 	& 	0		\\
    	0  	& 	1 		
\end{pmatrix}
$$
which governs the unitary part of the dynamics. Equation~\eqref{eqlindblad} is said to be in GKLS form when the different parameters and operators do no depend on time. The process is then described by a semigroup and a sufficient and necessary condition for CP is that all the rates $\gamma_i$ are positive. For two-level quantum systems, standard constraints are thus $\Gamma\geq \frac{\gamma_+}{2}$ and $\gamma_{12}\geq 0$, $\gamma_{21}\geq 0$, i.e.
\begin{equation}\label{eqmark}
\textrm{(M)} : 2\Gamma\geq\gamma_+\geq 0,~ -\gamma_+\leq\gamma_-\leq+\gamma_+
\end{equation}
In the general case, the rates $\gamma_i(t)$ may depend on time. A necessary and sufficient condition of CP-divisiblity for time-dependent rates is given by $\gamma_i(t)\geq 0$. The corresponding quantum processes are Markovian and CP, but they do not capture all the possible physical dynamics such as non-Markovian behaviors for which at least one of the rates $\gamma_i$ is negative during some time interval. In this case, a general characterization of CP or P has not yet been established and this is one of the objectives of this study to formulate such conditions for two-level open quantum systems. More precisely, the goal is to find constraints on the time-dependent decay rates generalizing the ones for the GKLS equation such that CP or P is verified.

We first recall results on the time evolution of the system~\eqref{eqdyntwo}, which can be integrated exactly in the coherence vector coordinates $(x,y,z)$~\cite{coherencevector}, defined for a two-level as $x=2\Re[\rho_{21}]$, $y=2\Im[\rho_{21}]$ and $z=\rho_{11}-\rho_{22}$. We denote by $(x_0,y_0,z_0)$ the initial values of the coherence vector coordinates at $t=0$. The dynamical system can then be expressed as
$$
\frac{d}{dt}\begin{pmatrix}
x \\
y \\
z
\end{pmatrix}
=\begin{pmatrix}
-\Gamma & \omega & 0\\
-\omega & -\Gamma & 0\\
0 & 0 & -\gamma_+
\end{pmatrix}
\begin{pmatrix}
x \\
y \\
z
\end{pmatrix}
+
\begin{pmatrix}
0 \\
0 \\
\gamma_-
\end{pmatrix}
$$
Introducing the coefficients $\tilde{\Gamma}=\int_0^t\Gamma(\tau)d\tau$, $\tilde{\gamma}_+=\int_0^t\gamma_+(\tau)d\tau$ and $\tilde{\omega}=\int_0^t\omega(\tau)d\tau$, straightforward computations lead to
$$
\begin{cases}
x(t)=e^{-\tilde{\Gamma}}(x_0\cos(\tilde{\omega})+y_0\sin(\tilde{\omega})) \\
y(t)=e^{-\tilde{\Gamma}}(-x_0\sin(\tilde{\omega})+y_0\cos(\tilde{\omega})) \\
z(t)=s(t)+z_0e^{-\tilde{\gamma}_+}
\end{cases}
$$
where $s$ is the particular solution of the differential equation satisfied by $z$, with $s(0)=0$. This solution can be expressed explicitly as a function of the decay rates as
$$
s(t)=e^{-\tilde{\gamma}_+(t)}\int_{0}^t[e^{\tilde{\gamma}_+(\tau)}\gamma_-(\tau)d\tau]
$$
Finally, we deduce that the dynamical map $\Phi_t$ can be written as
%\begin{widetext}
$$
\begin{pmatrix}
\rho_{11}(t) \\
\rho_{12}(t) \\
\rho_{21}(t) \\
\rho_{22}(t) \\
\end{pmatrix}
=
\begin{pmatrix}
\frac{1}{2}(1-s)+\frac{e^{-\tilde{\gamma}_+}}{2} & 0 & 0 & \frac{1}{2}(1-s)-\frac{e^{-\tilde{\gamma}_+}}{2}\\
0 & e^{-\tilde{\Gamma}-i\tilde{\omega}} & 0 & 0 \\
0 & 0 & e^{-\tilde{\Gamma}+i\tilde{\omega}} & 0\\
\frac{1}{2}(1+s)-\frac{e^{-\tilde{\gamma}_+}}{2} & 0 & 0 & \frac{1}{2}(1+s)+\frac{e^{-\tilde{\gamma}_+}}{2}  \\
\end{pmatrix}
\begin{pmatrix}
\rho_{11}(0) \\
\rho_{12}(0) \\
\rho_{21}(0) \\
\rho_{22}(0) \\
\end{pmatrix}
$$
%\end{widetext}
We denote by $\Phi_{ij}(t)$ the matrix elements of the map $\Phi_t$.

\section{Complete positivity of the dynamical map}\label{sec3}
Before establishing the conditions for CP of the quantum dynamics, standard mathematical results are briefly described.

We consider the map $\varphi$ from $M_n(\mathbb{C})$ to itself, where $M_n(\mathbb{C})$ is the set of $n\times n$-matrices with entries in $\mathbb{C}$. $\varphi$ is said to be positive if
$$
a\geq 0\  \Rightarrow \ \varphi(a)\geq 0
$$
with $a\in M_n(\mathbb{C})$. We recall that $a\geq 0$, i.e. $a$ is a positive semi-definite matrix, if for every complex vector $z$ we have $z^\dagger a z\geq 0$. Note that this condition implies that $a$ is a Hermitian matrix and that its eigenvalues are non-negative. A natural extension of the map $\varphi$ is $I_m\otimes \phi$ from $M_m(\mathbb{C})\otimes M_n(\mathbb{C}) \to M_m(\mathbb{C})\otimes M_n(\mathbb{C})$. $M_m(\mathbb{C})\otimes M_n(\mathbb{C})$ is identified as $m\times m$- matrices with entries in $M_n(\mathbb{C})$ and the map $I_m\otimes \phi$ is defined as
$$
\begin{pmatrix}
a_{11} & \cdots & a_{1m} \\
\vdots & \vdots & \vdots \\
a_{m1} & \cdots & a_{mm}
\end{pmatrix}
\mapsto
\begin{pmatrix}
\varphi(a_{11}) & \cdots & \varphi(a_{1m}) \\
\vdots & \vdots & \vdots \\
\varphi(a_{m1}) & \cdots & \varphi(a_{mm})
\end{pmatrix}
$$
where $a_{ij}\in M_n(\mathbb{C})$. By definition, $\varphi$ is CP if $I_m\otimes \phi$ is positive for all $m$. This property can be written explicitly as follows
$$
\sum_{i,j}^me_{ij}\otimes a_{ij}\geq 0\  \Rightarrow \ \sum_{i,j}^me_{ij}\otimes \varphi(a_{ij})\geq 0
$$
with $e_{ij}$ the matrix with one in the i-th row and j-th column and 0 elsewhere. We then introduce the Choi matrix $C_\varphi$ of $\varphi$ as
$$
C_\varphi=\sum_{i,j}^n e_{ij}\otimes \varphi(e_{ij})
$$
It can be shown that $\varphi$ is CP if and only if $C_\varphi$ is positive~\cite{choimatrix}.

In the quantum setting, the Choi matrix of $\Phi_t$ is given by~\cite{relaxrates1}
$$
C_{\Phi_t}=\sum_{i,j}|i\rangle\langle j|\otimes \Phi_t(|i\rangle\langle j|)
$$
Note that the Choi matrix can also be constructed from the maximally entangled state between two Hilbert spaces.

Using $\Phi_t(|i\rangle\langle j|)=\sum_{k,\ell}|k\rangle\langle k|\Phi_t(|i\rangle\langle j|)|\ell\rangle \langle\ell |$, we deduce that
$$
C_{\Phi_t}=\sum_{i,j}\sum_{k,\ell}|i,k\rangle\langle k|\Phi_t(|i\rangle\langle j|)|\ell\rangle\langle j,\ell|
$$
For a two-level quantum system, the matrix elements of the Choi matrix can be written as
$$
C_{\Phi_t}^{\alpha\beta} = \langle\alpha_1|\Phi_t[|\alpha_2\rangle\langle\beta_2|]|\beta_1\rangle,
$$
with $(\alpha_1,\alpha_2,\beta_1,\beta_2)\in \{1,2\}$. We introduce the operators $e_\alpha = |\alpha_1\rangle\langle\alpha_2|$ and $e_\beta = |\beta_1\rangle\langle\beta_2|$ and the basis
$e_1 = |1\rangle\langle 1|$, $e_2 = |1\rangle\langle 2|$, $e_3 = |2\rangle\langle 1|$ and $e_4 = |2\rangle\langle 2|$. In this basis, we then get
$$
C_{\Phi_t} =
\begin{pmatrix}
	\Phi_{11} 	&	0	&	0  &	\Phi_{22}	\\
    0	 	& 	\Phi_{14} 	& 	0	&	0\\
	0 		& 	0	& 	\Phi_{41}  	& 	0	\\		
	\Phi_{33}  	&	0	&	0  & \Phi_{44}		
\end{pmatrix}
$$
The expected properties of the Choi matrix can be verified, i.e. $C_{\Phi_t}$ is Hermitian and $\textrm{Tr}[C_{\Phi_t}]=2$~\cite{relaxrates1}. The characteristic polynomial of the matrix is $p(X) = (\Phi_{14} - X)(\Phi_{41} - X)((\Phi_{11} - X)(\Phi_{44} - X) - \Phi_{22}\Phi_{33}) = (\Phi_{14} - X)(\Phi_{41} - X)(X^2 - (\Phi_{11} + \Phi_{44})X + \Phi_{11}\Phi_{44} - \Phi_{22}\Phi_{33})$. We have $\Phi_{11} + \Phi_{44} = 1 + e^{-\tilde{\gamma}_+}$, $\Phi_{11}\Phi_{44} = \frac{1}{4}((1 + e^{-\tilde{\gamma}_+})^2 - s^2)$ and $\Phi_{22}\Phi_{33} = e^{-2\tilde{\Gamma}}$. The eigenvalues of $C_{\Phi_t}$ are $\Phi_{14}$, $\Phi_{41}$ and $\frac{1}{2}\left(\Phi_{11} + \Phi_{44} \pm \sqrt{(\Phi_{11} - \Phi_{44})^2 + 4\Phi_{22}\Phi_{33}}\right)$. The Choi matrix is positive if all the eigenvalues are positive. We then deduce that the necessary and sufficient conditions of CP can be expressed for all times $t\geq 0$ as
\begin{eqnarray}
\label{eq1} & & -(1 - e^{-\tilde{\gamma}_+}) \leq s \leq 1 - e^{-\tilde{\gamma}_+}\\
\label{eq2} & & s^2 \leq (1 + e^{-\tilde{\gamma}_+})^2 - 4e^{-2\tilde{\Gamma}}
\end{eqnarray}
%We also point out that the Choi matrix can be constructed from the maximally entangled state $|\psi_M\rangle$ as follows. This state can be expressed as
%$$
%|\psi_M\rangle=\frac{1}{\sqrt{N}}\sum_i|i,i\rangle.
%$$
%for an entangled state between two identical Hilbert spaces of dimension $N$ and basis $\{|i\rangle\}$. We deduce that the corresponding density operator $\rho_M=|\psi_M\rangle\langle \psi_M |$ is given by
%\begin{eqnarray*}
%\rho_M &=&\frac{1}{N}\sum_{i,j}|i,i\rangle\langle j,j|=\frac{1}{N}\sum_{i,j}|i\rangle\langle j|\otimes |i\rangle\langle j|\\
%&=& \frac{1}{N}\sum_{i,j}e_{ij}\otimes e_{ij}
%\end{eqnarray*}
%Up to a factor $N$, the Choi matrix is the matrix representation of $(I_n\otimes \Phi)[\rho_M]$.
Note that Eq.~\eqref{eq1} implies that $\tilde{\gamma}_+(t)\geq 0$ for $t\geq0$, while Eq.~\eqref{eq2} leads to $\tilde{\Gamma}(t)\geq0$. Indeed, we have $0\leq (1 + e^{-\tilde{\gamma}_+})^2 - 4e^{-2\tilde{\Gamma}}$, i.e.
$e^{-2\tilde{\Gamma}}\leq \frac{1}{4}(1 + e^{-\tilde{\gamma}_+})^2\leq 1$.
The CP conditions are then equivalent to
\begin{eqnarray}
\label{eq1new} & & \tilde{\gamma}_+(t)\geq 0,~ \tilde{\Gamma}(t)\geq0,~\forall t\geq0\nonumber \\
\label{eq2new} \textrm{(CP)}:~& & s^2 \leq (1 - e^{-\tilde{\gamma}_+})^2~ \textrm{if}\ \tilde{\gamma}_+ \leq 2\tilde{\Gamma}\\
\label{eq3new} & & s^2 \leq (1 + e^{-\tilde{\gamma}_+})^2 - 4e^{-2\tilde{\Gamma}}~\textrm{if}\ \tilde{\gamma}_+ \geq 2\tilde{\Gamma}
\end{eqnarray}
Indeed, it is straightforward to see that if $\tilde{\gamma}_+(t) \leq 2\tilde{\Gamma}(t)$ then (\ref{eq2new}) $\Rightarrow$ (\ref{eq3new}) at time $t$, whereas (\ref{eq3new}) $\Rightarrow$ (\ref{eq2new}) at time $t$ in the case $\tilde{\gamma}_+(t) \geq 2\tilde{\Gamma}(t)$. Contrary to the criteria of Markovianity and non-Markovianity, we also observe that these constraints depend on the time integral of the relaxation rates.
\section{Positivity of the dynamical map}\label{sec4}
We study in this section under which conditions $\Phi_t$ is a positive map. By definition, we already know that CP is a stronger condition than P. It is nevertheless instructive to prove this result in the case studied. We then establish necessary and sufficient conditions for P that can be directly compared to those of CP. Using such conditions, we describe a family of dynamical maps for which CP and P are equivalent and we exhibit examples where the map is only positive. We recall that a positive map transforms a positive operator into a positive operator. Starting from a density operator $\rho_0$ at time 0, the question is thus to find under which conditions the state at time $t$ is also a well-defined density operator. In the coherence vector coordinates, this corresponds to the condition $x(t)^2+y(t)^2+z(t)^2\leq 1$ at any time $t$, knowing that $x_0(t)^2+y_0(t)^2+z_0(t)^2\leq 1$.
\begin{theorem}\label{th1}
A CP map $\Phi_t$ is positive.
\end{theorem}
\begin{proof}
We have
\begin{eqnarray*}
x(t)^2 + y(t)^2 + z(t)^2 &=& e^{-2\tilde{\Gamma}(t)}(x_0^2 + y_0^2) + z(t)^2\\
&\leq& e^{-2\tilde{\Gamma}(t)}(1 - z_0^2) + z(t)^2.
\end{eqnarray*}
The equality is reached when $x_0^2 + y_0^2 + z_0^2 = 1$. Thus, the system is positive if and only if $e^{-2\tilde{\Gamma}}(1 - z_0^2) + z^2 \leq 1$. If $\tilde{\gamma}_+ \leq 2\tilde{\Gamma}$ then we have
$$
e^{-2\tilde{\Gamma}}(1 - z_0^2) + z^2 \leq e^{-\tilde{\gamma}_+}(1 - z_0^2) + z^2.
$$
Using Eq.~(\ref{eq2new}), we get $-(1 - e^{-\tilde{\gamma}_+}) \leq s(t) \leq 1 - e^{-\tilde{\gamma}_+}$,
which leads to
$$
-(1 - (1 + z_0)e^{-\tilde{\gamma}_+}) \leq z(t) \leq 1 - (1 - z_0)e^{-\tilde{\gamma}_+}.
$$
Since
\begin{eqnarray*}
& & \left(1 - (1 - z_0)e^{-\tilde{\gamma}_+(t)}\right)^2 - \left(1 - (1 + z_0)e^{-\tilde{\gamma}_+(t)}\right)^2 \\
& & = 4e^{-\tilde{\gamma}_+(t)}(1 - e^{-\tilde{\gamma}_+(t)})z_0,
\end{eqnarray*}
we deduce that the maximum of $z^2$ is $\left(1 - (1 - z_0)e^{-\tilde{\gamma}_+(t)}\right)^2$ if $z_0 \geq 0$ and $ \left(1 - (1 + z_0)e^{-\tilde{\gamma}_+(t)}\right)^2$ if $z_0 \leq 0$. Let us assume that $z_0 \geq 0$, the case $z_0\leq 0$ can be done along the same lines. We have
\begin{eqnarray*}
e^{-2\tilde{\Gamma}}(1 - z_0^2) + z^2 &\leq& e^{-\tilde{\gamma}_+}(1 - z_0^2) + z^2\\
&\leq& e^{-\tilde{\gamma}_+}(1 - z_0^2) + \left(1 - (1 - z_0)e^{-\tilde{\gamma}_+}\right)^2\\
&\leq& - e^{-\tilde{\gamma}_+}(1 -  e^{-\tilde{\gamma}_+})\left(z_0 - 1\right)^2 + 1\\
&\leq&  1,
\end{eqnarray*}
leading to $x(t)^2 + y(t)^2 + z(t)^2\leq 1$ for any $t\geq0$.

Consider now that $\tilde{\gamma}_+ \geq 2\tilde{\Gamma}$. From Eq.~(\ref{eq3new}), we have
$$
s^2 \leq (1 + e^{-\check{\gamma}_+})^2 - 4e^{-2\check{\Gamma}} \leq (1 + e^{-2\check{\Gamma}})^2 - 4e^{-2\check{\Gamma}} = (1 - e^{-2\check{\Gamma}})^2.
$$
Since $\tilde{\Gamma}\geq 0$, we get
$$
-(1 - e^{-2\tilde{\Gamma}}) \leq s(t) \leq 1 - e^{-2\tilde{\Gamma}}.
$$
The proof is thus the same as in the first case, replacing $2\tilde{\Gamma}$ by $\tilde{\gamma}_+$.\qed
\end{proof}
As for CP, necessary and sufficient conditions for P can be established. The dynamical map $\Phi_t$ is P if and only if
\begin{eqnarray}
\label{cn1} & & \tilde{\gamma}_+(t) \geq 0, \tilde{\Gamma}(t)\geq 0,~\forall t\geq0 \nonumber \\
\label{cn2} \textrm{(P)}: & & s^2 \leq (1 - e^{-\tilde{\gamma}_+})^2~ \textrm{if}\ \tilde{\gamma}_+ \leq 2\tilde{\Gamma}\\
\label{cn3} & & s^2 \leq (1 - e^{-2\tilde{\Gamma}})(1 - e^{-2(\tilde{\gamma}_+ - \tilde{\Gamma})}) ~\textrm{if}\ \tilde{\gamma}_+ \geq 2\tilde{\Gamma}
\end{eqnarray}
These conditions are proved in Appendix~\ref{appb}. Since the criteria~\eqref{eq2new} and~\eqref{cn2} are the same, it is straightforward to show that if $0 \leq \tilde{\gamma}_+ \leq 2\tilde{\Gamma}$, then CP is equivalent to P. A family of maps for which the two properties are not equivalent can be found using the conditions~\eqref{eq3new} and~\eqref{cn3}.
For instance, if $\tilde{\Gamma}(t) = 0, \tilde{\gamma}_+(t) > 0$ and $\gamma_-(t) = 0$, then the dynamic is P but not CP.
Indeed, since $\gamma_- = 0$ we have $s = 0$ and the condition~\eqref{cn3} is satisfied. On the contrary, it is not CP because the condition~\eqref{eq3new}, $s^2 \leq (1 + e^{-\tilde{\gamma}_+})^2 - 4e^{-2\tilde{\Gamma}}$ writes $0 \leq (1 + e^{-\tilde{\gamma}_+})^2 - 4$ and is not satisfied.
\section{Quasi-Markovian systems}\label{sec5}
As already mentioned, a Markovian system for which the decay rates $\gamma_i(t)$ of the GKLS equation are positive for any time $t$ is CP. An interesting issue is to find similar conditions for non-Markovian systems since any NM dynamic is not CP. A first answer has been given in Sec.~\ref{sec3} with some explicit conditions for CP in a two-level open quantum system. However, this solution is not completely satisfactory in the sense that such conditions are not easy to check quickly or to interpret physically. We propose in this paragraph a new family of dynamical systems, called the quasi-Markovian systems, which is larger than the Markovian one and for which the dynamical map is CP. The conditions of QM are directly inspired from those of a Markovian system, except that a time-local condition is replaced by a time-integral one. A two-level quantum system is said to be \emph{quasi-Markovian} (QM) if the following inequalities are verified for all $t\geq0$
$$
\textrm{(QM)}:~-|\gamma_+| \leq \gamma_- \leq +|\gamma_+|~\textrm{and}~\tilde{\Gamma} \geq \frac{1}{2}\tilde{\gamma}_+\geq0.
$$
Using Eq.~\eqref{eqmark}, we obtain that the Markovian systems are QM. It can then be shown that the dynamical map of a QM system is CP.
Indeed, as shown in Appendix~\ref{appC}, a QM system fulfills
\begin{equation}\label{eqQM}
-(1 - e^{-\tilde{\gamma}_+}) \leq s(t) \leq 1 - e^{-\tilde{\gamma}_+}.
\end{equation}
We deduce that $s^2\leq (1 - e^{-\tilde{\gamma}_+})^2 = (1 + e^{-\tilde{\gamma}_+})^2 - 4e^{-\tilde{\gamma}_+}$ and then from Eq.~\eqref{eq2new} we obtain the CP.

This result is interesting because it allows to find easily CP maps which are also NM. It is the case for instance if the function $\Gamma(t)$ takes negative values while satisfying  $\tilde{\Gamma}\geq\tilde{\gamma}_+/2$ for all $t\geq0$. Different examples will be given in Sec.~\ref{sec6}.
The different properties of the dynamical map are summarized in Fig.~\ref{fig0}.
\begin{figure}[tb]
  \centering
  \includegraphics[width=1\linewidth]{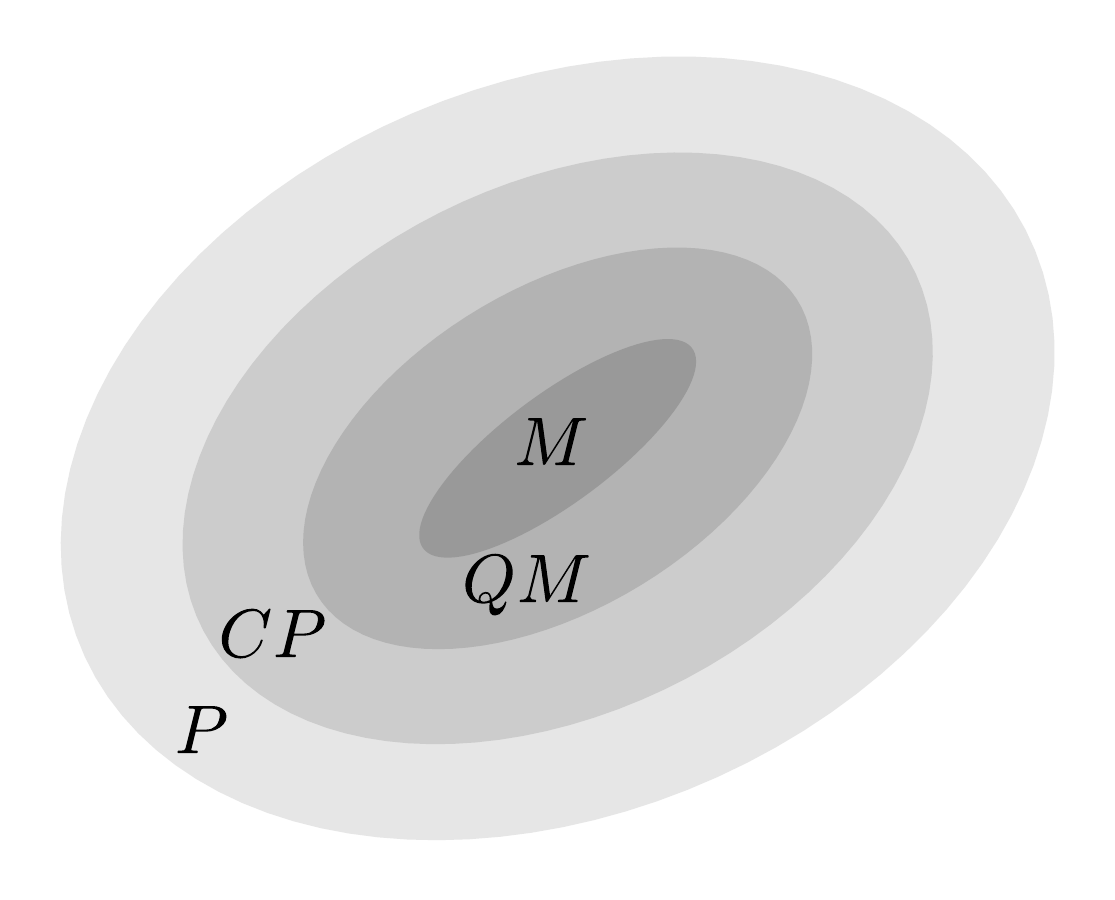}
   \caption{Schematic representation of the different characteristics of the dynamical map for which $M\Rightarrow QM \Rightarrow CP \Rightarrow P$. The different gray areas correspond to larger and larger sets of dynamical maps ranging from Markovian maps to positive ones.}
  \label{fig0}
\end{figure}
\section{Long-time asymptotic behavior}\label{secnew}
We have characterized in the preceding sections the system dynamics. Another key point is to describe the long-time asymptotic behavior of the density matrix. In particular, the goal is to establish under which conditions on the decay rates the system tends to an equilibrium state. This problem is quite simple in the Markovian regime with constant coefficients, but it is much more complex for time-dependent rates and non-Markovian dynamics. We introduce below a class of functions which ensure the existence of this asymptotic state and we describe in which cases this state can be explicitly found. This analysis is interesting in the design of a quantum master equation since it helps to select the right family of time-dependent relaxation rates to consider.

Different results can be established for the differential equation (E): $\dot{z} = Az + B$ where $A,B$ are two time-dependent functions converging to $A_0,B_0$, respectively, when $t\to +\infty$. The different proofs are given in Appendix~\ref{secasympnew}.
\begin{proposition}\label{prop4}
If $A_0\neq 0$ then any bounded solution $z$ of (E) converges to $-\frac{B_0}{A_0}$.
\end{proposition}
The case $A_0 = 0$ is different. Indeed, nonconvergent bounded solutions of specific equation $(E)$ can be found. An example is given by $z(t) = \sin(\sqrt{t})$. We have $\dot{z}(t) = \frac{1}{2\sqrt{t}}\cos(\sqrt{t}) = \frac{1}{2\sqrt{t}}\sin(\sqrt{t}) + \frac{1}{\sqrt{2t}}\cos(\sqrt{t} + \frac{\pi}{4})$.
We deduce that $z$ is a solution of the differential equation $\dot{z} = Az + B$ with $A,B\to 0$, but $z$ is a nonconvergent bounded function.
The non convergent behavior is due to the slow oscillations of $z(t) = \sin(\sqrt{t})$. This justify the following definition.
\begin{definition}
A function $f$ is said to be slowly oscillating if for all $\tau>0$, $f(t+\tau) - f(t)$ goes to 0 when $t\to +\infty$.
\end{definition}
This means that such a function looks more and more like its time-shifted version as $t$ tends to infinity. It can be shown that if $z$ is a bounded solution of (E) with $A,B$ converging to 0, then $z$ is a slowly oscillating function.
Since a slowly oscillating function is not necessary convergent, additional conditions on $A$ and $B$ are required to ensure the convergence of $z$.
\begin{definition}
A function $f$ goes to $0$ not too slowly if there exists $\alpha>0$ such that $\lim_{t\to+\infty} t^{1+\alpha}f(t) = 0$ i.e. $f(t) \in o(\frac{1}{t^{1+\alpha}})$.
\end{definition}
The following result can then be proved.
\begin{proposition}\label{prop5}
Let $z$ be a bounded solution of (E) where $A,B$ are two functions going not too slowly to 0. Then $z$ is a convergent function.
\end{proposition}
Note that $z$ can have any real limit. Consider for instance the differential equation $\dot{z} = \frac{1}{t^2}z + \frac{1}{t^2}$.
The solution can be expressed as
$$
z(t) = K e^{-\frac{1}{t}} - 1,\qquad K\in\mathbb{R}.
$$
These functions converge and the limit is $K-1$.

Finally, we come back to the dynamical system satisfied by the coherence vector
$$
\begin{cases}
\dot{x}=-\Gamma x+\omega y \\
\dot{y}=-\Gamma y-\omega x \\
\dot{z}=\gamma_--\gamma_+z
\end{cases}
$$
We denote by $\gamma_-^0$, $\gamma_+^0$ and $\Gamma^0$ the limits of the different decay rates when $t\to +\infty$. Physically, the equilibrium state is usually defined as the coherence vector of coordinates $(0,0,\gamma_-^0/\gamma_+^0)$ when $\gamma_+^0\neq 0$. We consider a CP dynamical map for which $x(t)^2+y(t)^2+z(t)^2\leq 1$ for all $t\geq 0$. It is straightforward to see that $x$, $y$ and $z$ are bounded functions. Using~Eq.~\eqref{eq1new}, we know that $\tilde{\gamma}_+\geq 0$ and $\tilde{\Gamma}\geq 0$ and we deduce that $\gamma_+^0\geq 0$ and $\Gamma^0\geq 0$. The different results of this section can then be used.

When $\gamma_+^0> 0$ and $\Gamma^0> 0$, we have from Prop.~\ref{prop4} that the coherence vector goes to $(0,0,\gamma_-^0/\gamma_+^0)$ when $t\to +\infty$. In the case $\gamma_\pm^0=0$ or $\Gamma^0=0$, from Prop.~\ref{prop5}, different limits (if they exist) can be obtained according to the functions $\gamma_+(t),\gamma_-(t)$ and $\Gamma(t)$. The convergence is ensured if the three functions are not too slowly oscillating functions.
\section{Numerical examples}\label{sec6}
We first consider the general case of a qubit with multiple decoherence channels~\cite{RMPbreuer,measureother5,vacchini2012}. The master equation can be written as
$$
\dot{\rho}=\sum_i\frac{\gamma_i}{2}(\sigma_i\rho\sigma_i-\rho)
$$
where the $\sigma_i$ are the Pauli matrices and $i=x,y,z$. This equation can be expressed in terms of the coherence vector coordinates as follows
$$
\begin{cases}
\dot{x}=-(\gamma_y+\gamma_z)x \\
\dot{y}=-(\gamma_x+\gamma_z)y \\
\dot{z}=-(\gamma_x+\gamma_y)z
\end{cases}
$$
{\red This model system can be viewed as an empirical model describing the dynamics of a qubit in a complex environment. The rates $\gamma_i(t)$ can be associated with transverse and longitudinal rates generalizing the standard $1/T_1$ and $1/T_2$ constant rates used to describe a dissipative spin 1/2 particle in magnetic resonance~\cite{lapert2013}.} We assume that $\gamma_x=\gamma_y$ so as to be in the conditions of this study, but a similar analysis could be done for different decay rates. We introduce the coefficients $\tilde{\gamma}_+=\int_0^t[\gamma_x+\gamma_y]d\tau$ and $\tilde{\Gamma}=\int_0^t[\gamma_x+\gamma_z]d\tau$. Since $\gamma_-=0$, we deduce that $s=0$. The only condition to satisfy is thus $1+e^{-\tilde{\gamma}_+}\geq 2e^{-\tilde{\Gamma}}$.
As a specific example, we consider the case of eternal non-Markovianity~\cite{relaxrates1,megier2017} for which $\gamma_x=\gamma_y=1$ and $\gamma_z=-\tanh t$. Note that this system is NM for all times $t$ since $\gamma_z(t)<0$. We get $\gamma_+=2$, $\gamma_-=0$ and $\Gamma=1-\tanh t$ which leads to
$$
\tilde{\gamma}_+=2t,~\tilde{\Gamma}=\ln(\frac{e^t}{\cosh t})
$$
We can verify that $\tilde{\Gamma}< \frac{\tilde{\gamma}_+}{2}$ at any time $t>0$, so the system is never quasi-Markovian. The CP can be verified from our criterion. We have $s(t)=0$ and the relation \eqref{eq3new} has to be fulfilled as
\begin{eqnarray*}
(1+e^{-\tilde{\gamma}_+})^2-4e^{-2\tilde{\Gamma}}&=& 1+e^{-4t}+2e^{-2t}-4e^{-2t}\cosh^2 t\\
& \geq & (1-e^{-2t})^2
\end{eqnarray*}
Another standard example is a two-level system coupled to a lossy cavity~\cite{RMPbreuer,breuerbook}. This system corresponds, e.g. to a single two-level atom interacting with an electromagnetic field having a Lorenztian spectral density which mimics a lossy cavity. The dynamics of the coherence vector are given by
$$
\begin{cases}
\dot{x}=-\frac{\gamma}{2}x+\frac{S}{2}y \\
\dot{y}=-\frac{\gamma}{2}y-\frac{S}{2}x \\
\dot{z}=\gamma-\gamma z
\end{cases}
$$
where $\gamma$ and $S$ are two time-dependent functions. We deduce that
$$
\gamma_-=\gamma_+=\gamma,~\Gamma=\frac{\gamma}{2}
$$
This system is QM because $-|\gamma|\leq \gamma\leq |\gamma|$. We verify also the CP criterion. The particular solution is $s(t)=1-e^{-\tilde{\gamma}}$. It is then straightforward to show that
$$
s^2=(1-e^{-\tilde{\gamma}})^2,
$$
i.e. the condition \eqref{eq2new}.

As a final example, we consider a model system depending on a parameter and we study the dynamical properties as a function of this parameter. {\red Note that this example does not correspond in itself to a physical system, but it is interesting to study to highlight the transition between the different behaviors.} We consider the functions $\gamma_+ = 1$ and $\gamma_- = \alpha(1 + e^{-t})$ with $0\leq\alpha\leq1$. For the rate $\Gamma$, we choose any function such that $\tilde{\Gamma}\geq t/2$. Since $\gamma_-$ is a strictly decreasing function from $2\alpha$ to $\alpha$, the system is QM if and only if $\alpha \in [0,\frac{1}{2}]$. The system is Markovian when $\Gamma\geq0$, but non-Markovian examples can be found. For $\alpha<1$, the system is non-Markovian at short times and eternally non-Markovian if $\alpha = 1$. The differential system for the $z$- coordinate is
$$
\dot{z} = -z + \alpha(1 + e^{-t})
$$
and the solutions can be expressed as
$$
z(t) = \alpha(1 - (1 - t)e^{-t}) + Ae^{-t}.
$$
with a constant $A$. We deduce that $s(t) = \alpha(1 - (1 - t)e^{-t})$. This function is increasing from 0 to $\alpha(1 + e^{-2})$ on the interval $[0,2]$ and then decreasing to $\alpha$ on $[2,+\infty[$. We consider now the CP conditions. Only the condtion \eqref{eq2new} has to be verified. We have
$$
1 - e^{-\tilde{\gamma}_+} \pm s \geq 0 \iff 1 \pm \alpha - e^{-t}(1 \pm \alpha(1-t))\geq0.
$$
The derivative of the function $f_\pm(t) = 1 \pm \alpha - e^{-t}(1 \pm \alpha(1-t))$ has the same sign as $\mp t \pm 2 + 1/\alpha$. We have $f_\pm (0) = 0$.
The function  $f_+$ is strictly increasing on $]-\infty,2+1/\alpha]$ and strictly decreasing on $[2+1/\alpha,+\infty[$ with a limit equal to $1+\alpha$.
We deduce that $f_+ \geq 0$ on $[0,+\infty[$. The function $f_-$ is strictly decreasing on $]-\infty,2-1/\alpha]$ and strictly increasing on $[2-1/\alpha,+\infty[$ with a limit equal to $1-\alpha\geq 0$. The minimum is $1 - \alpha(1 + e^{-2+1/\alpha})$. Since $2 - 1/\alpha\geq0$ if and only if $\alpha \geq 1/2$, the minimum of $f_-$ is 0 when $\alpha \leq 1/2$ and negative equal to $1 - \alpha(1 + e^{-2+1/\alpha})$ otherwise. In conclusion, the dynamical map is CP in the QM case with a global maximum for $s$.
\begin{figure}[tb]
  \centering
  \includegraphics[width=1\linewidth]{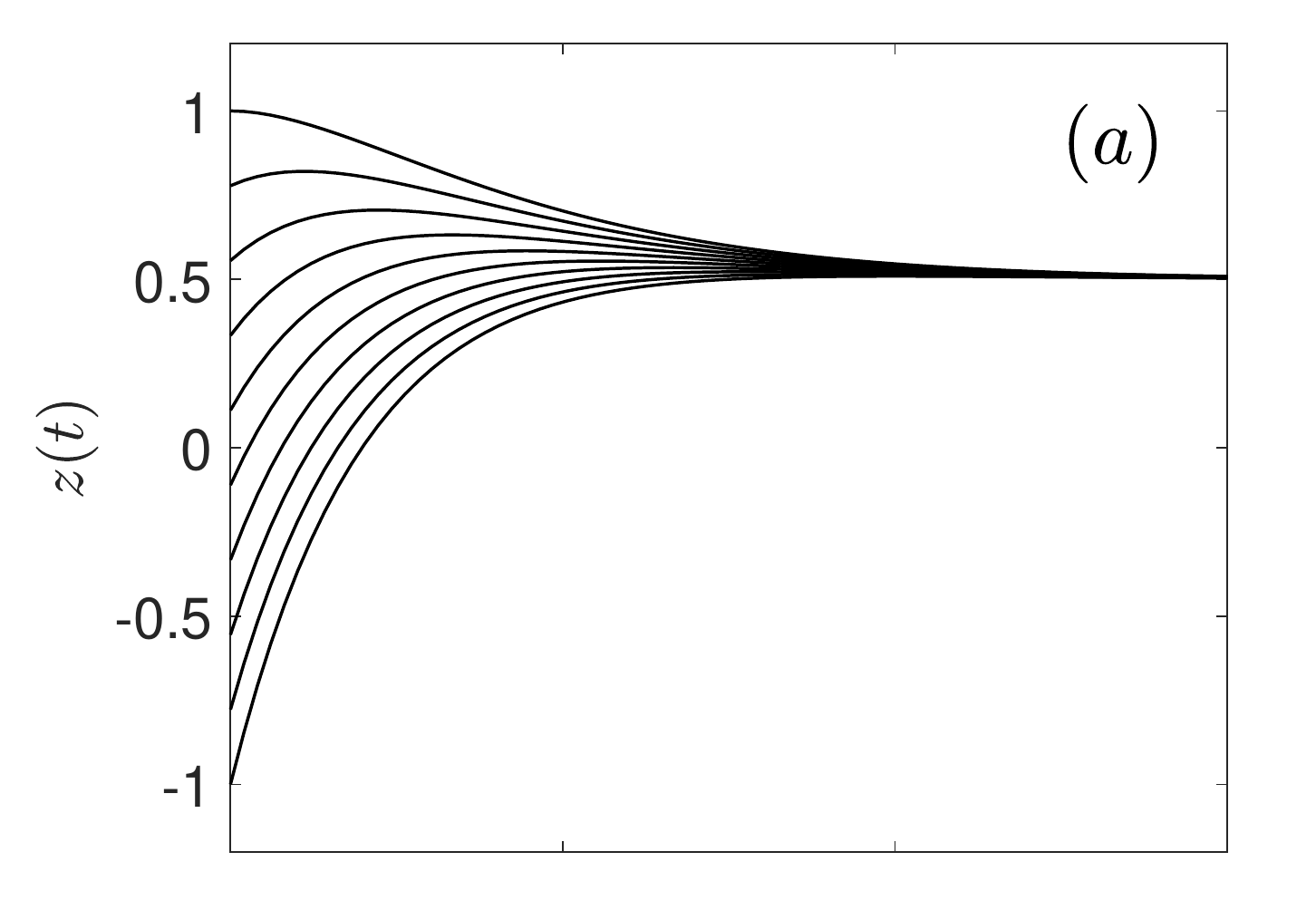}
  \includegraphics[width=1\linewidth]{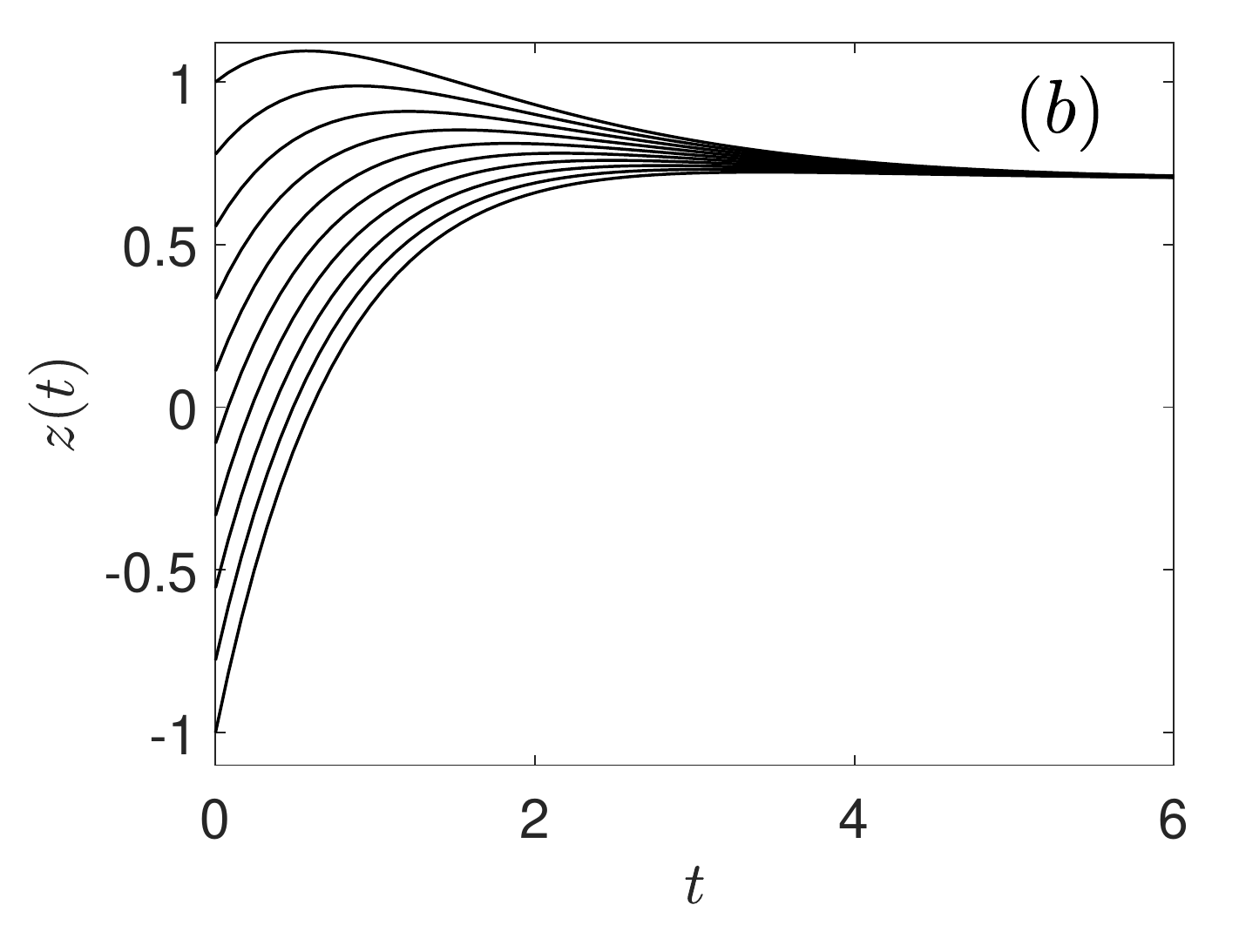}
  \caption{Plot of the trajectories $z(t)$ for a parameter $A$ going from -1 to 1 (from bottom to top). The parameter $\alpha$ is set respectively to 0.5 and 0.7 in the panels (a) and (b).}
  \label{fig1}
\end{figure}
However, for a NM system, the system is no more positive. The coherence vector goes out from the sphere of radius one for specific initial conditions. Indeed, for $A<0.8$, the coordinate $z(t)$ belongs to the interval $[-1,1]$, but this property is not verified when $A \simeq \pm1$. The different trajectories for the two behaviors are represented in Fig.~\ref{fig1}.
\section{Conclusion}\label{sec7}
We have investigated the concepts of CP and P in a two-level open quantum system whose dynamics are governed by a time-local master equation. Assuming that the decay rates can be represented by smooth time-dependent functions, we have established different criteria about the CP and the P of the dynamical map. A subtlety of such conditions is that they are nonlocal in time, in the sense that the criteria are established for the time integral of the decay rates. It is a major difference with the different measures of non-Markovianity which are local and only involve these rates at a given time. This observation partly explains the difficulty to establish simple and general conditions of CP in non-Markovian systems, while such criteria exist for Markovian ones. In the second part of this study, we simplify this question by introducing the concept of quasi-Markovianity, which corresponds to a larger class of systems than the Markovian systems, but does not include all possible non-Markovian dynamics. Interestingly, these systems are characterized by local and non-local conditions on decay rates. QM allows to design quite easily examples which are both NM and CP. The question of generalizing quasi-Markovian systems to a larger class of non-Markovian systems would be interesting to study. Finally, we show under which conditions on the relaxation rates a two-level quantum system characterized by a CP dynamical map tends asymptotically to its equilibrium state.

This study paves also the way to other promising issues. An interesting question is to study the link between the conditions obtained in this study for CP and fundamental thermodynamic principles used for deriving time-local master equation as described, e.g., in~\cite{dann1,dann2}. A next problem consists in investigating controlled open dynamics with relaxation rates which may depend on the control parameters~\cite{glaserreview,roadmap,kochreview,kochRMP}. A similar analysis could establish conditions on the decay rates allowing to preserve the CP of the dynamical map against control variations. The impact of an external control on NM has been already investigated in a series of works~\cite{lapert2013,mangaud2018,reich2015,poggi2017,mukherjee2015}.

\noindent\textbf{Acknowledgment}\\
This research has been supported by the ANR project ``QuCoBEC'' ANR-22-CE47-0008-02. We gratefully acknowledge useful discussions with R. Kosloff and R. Dann.\\

\appendix

\section{Necessary and sufficient conditions for the positivity of the dynamical map}\label{appb}
We prove in this section the necessary and sufficient conditions \eqref{cn2} and \eqref{cn3} on the positivity of the dynamical map.

We consider the initial condition $x_0 = y_0 = 0$ and $z_0 = \pm 1$. We get
\begin{equation*}
\label{CN1} 0 \leq \left(s(t) \pm e^{-\tilde{\gamma}_+(t)}\right)^2 \leq 1,
\end{equation*}
which leads to $0 \leq s^2(t) + e^{-2\tilde{\gamma}_+(t)} \leq 1$. We deduce that $e^{-2\tilde{\gamma}_+(t)} \leq 1$ and $\tilde{\gamma}_+(t) \geq 0$.

From the condition $x_0^2 + y_0^2 = 1$ and $z_0 = 0$, we arrive at $e^{-2\tilde{\Gamma}(t)} + s^2(t) \leq 1$ and $e^{-2\tilde{\Gamma}(t)} \leq 1$ i.e. $\tilde{\Gamma}(t) \geq 0$.

The next step consists in showing that the inequality $x(t)^2 + y(t)^2 + z(t)^2 \leq 1$ can be rewritten as
$$
e^{-2\tilde{\Gamma}(t)}(1 - z_0^2) + \left(s(t) + z_0e^{-\tilde{\gamma}_+(t)}\right)^2 - 1 \leq 0
$$
or
$$
z_0^2(e^{-2\tilde{\gamma}_+(t)} - e^{-2\tilde{\Gamma}(t)}) + 2z_0s(t)e^{-\tilde{\gamma}_+(t)} +s(t)^2 - 1 + e^{-2\tilde{\Gamma}(t)} \leq 0
$$
with the condition $x_0^2 + y_0^2 + z_0^2 = 1$. We introduce the function $Q(z_0) = z_0^2(e^{-2\tilde{\gamma}_+(t)} - e^{-2\tilde{\Gamma}(t)}) + 2z_0s(t)e^{-\tilde{\gamma}_+(t)} + s(t)^2 - 1 + e^{-2\tilde{\Gamma}(t)}$.

\noindent
\textbf{Case $\tilde{\gamma}_+ \leq \tilde{\Gamma}$.}
A necessary and sufficient condition for P is $Q(-1) \leq 0$ and $Q(1) \leq 0$ (the coefficient of the higher degree term being positive). We obtain
\begin{eqnarray*}
& & -1 - e^{-\tilde{\gamma}_+(t)} \leq s(t) \leq 1 - e^{-\tilde{\gamma}_+(t)}\\
& & -1 + e^{-\tilde{\gamma}_+(t)} \leq s(t) \leq 1 + e^{-\tilde{\gamma}_+(t)}.
\end{eqnarray*}
Using $\tilde{\gamma}_+(t) \geq 0$, we get $-1 + e^{-\tilde{\gamma}_+(t)} \leq s(t) \leq 1 - e^{-\tilde{\gamma}_+(t)},
$ or $s^2 \leq (1 - e^{-\tilde{\gamma}_+})^2$, i.e. the condition~\eqref{cn2} when $\tilde{\gamma}_+ \leq \tilde{\Gamma}$.\\
\noindent
\textbf{Case $\tilde{\Gamma} < \tilde{\gamma}_+$.}
The coefficient of the higher degree term of the polynomial $Q$ is negative. We denote by $\tilde{z}_0$ the coordinate giving the maximum of $Q(z_0)$.
Since
$$
\tilde{z}_0 = \frac{s(t)e^{-\tilde{\gamma}_+(t)}}{e^{-2\tilde{\Gamma}(t)} - e^{-2\tilde{\gamma}_+(t)}},
$$
we deduce that the positivity is equivalent to
\begin{eqnarray*}
& &  Q(\tilde{z}_0)\leq 0~\textrm{if}~-1 \leq \tilde{z_0} \leq +1 \\
& & Q(+1)\leq 0~\textrm{if}~\tilde{z}_0\geq+1 \\
& & Q(-1)\leq 0~\textrm{if}~\tilde{z}_0\leq -1
\end{eqnarray*}
We then consider two subcases.

Assume that $2\tilde{\Gamma} < \tilde{\gamma}_+$. When $\tilde{z}_0\geq+1$, we have $s(t) \geq e^{\tilde{\gamma}_+(t)}(e^{-2\tilde{\Gamma}(t)} - e^{-2\tilde{\gamma}_+(t)})\geq0$.
The condition $Q(+1)\leq 0$ leads to $s\leq 1 - e^{-\tilde{\gamma}_+(t)}$. We get therefore
\begin{align*}
e^{\tilde{\gamma}_+(t)}(e^{-2\tilde{\Gamma}(t)} - e^{-2\tilde{\gamma}_+(t)}) &\leq 1 - e^{-\tilde{\gamma}_+(t)}\\
e^{\tilde{\gamma}_+(t) - 2\tilde{\Gamma}(t)}  &\leq 1 \\
\tilde{\gamma}_+(t) - 2\tilde{\Gamma}(t) &\leq 0,
\end{align*}
so a contradiction. When $\tilde{z}_0\leq-1$, we have $s(t) \leq -e^{\tilde{\gamma}_+(t)}(e^{-2\tilde{\Gamma}(t)} - e^{-2\tilde{\gamma}_+(t)})\leq0$ and the condition $Q(-1)\leq 0$ gives $s\geq -1 + e^{-\check{\gamma}_+(t)}$. We deduce that
\begin{align*}
-e^{\tilde{\gamma}_+(t)}(e^{-2\tilde{\Gamma}(t)} - e^{-2\tilde{\gamma}_+(t)}) &\geq -1 + e^{-\tilde{\gamma}_+(t)}\\
e^{\tilde{\gamma}_+(t) - 2\tilde{\Gamma}(t)}  &\leq 1 \\
\tilde{\gamma}_+(t) - 2\tilde{\Gamma}(t) &\leq 0,
\end{align*}
and a contradiction. Finally, when $2\tilde{\Gamma} < \tilde{\gamma}_+$, the necessary and sufficient condition is $Q(\tilde{z}_0)\leq 0$, which gives
$$
s^2 \leq (1 - e^{-2\tilde{\Gamma}(t)})(1 - e^{-2(\tilde{\gamma}_+(t) - \tilde{\Gamma}(t))}),
$$
i.e. the condition~\eqref{cn3}.

Assume now that $\tilde{\Gamma} < \tilde{\gamma}_+ \leq 2\tilde{\Gamma}$. There is no obstruction and the necessary conditions are sufficient.
The second condition~\eqref{cn3} implies the condition~\eqref{cn2}. We have
\begin{eqnarray*}
& & (1 - e^{-2\tilde{\Gamma}})(1 - e^{-2(\tilde{\gamma}_+ - \tilde{\Gamma})}) - (1 - e^{-\tilde{\gamma}_+})^2 \\
& & \leq (1 - e^{-\tilde{\gamma}_+})(1 - e^{-2(\tilde{\gamma}_+ - \tilde{\Gamma})}) - (1 - e^{-\check{\gamma}_+})^2\\
& & \leq (1 - e^{-\tilde{\gamma}_+})(e^{-\tilde{\gamma}_+} - e^{-2(\tilde{\gamma}_+ - \tilde{\Gamma})})\\
& & \leq e^{-\tilde{\gamma}_+}(1 - e^{-\tilde{\gamma}_+})(1 - e^{2\tilde{\Gamma} - \tilde{\gamma}_+})\leq0,
\end{eqnarray*}
which gives the result.
\section{Properties of quasi-Markovian systems}\label{appC}
We show in this section different results used to prove Eq.~\eqref{eqQM}.
\begin{lemma}
\label{lem_1}
Let $\gamma_-,\gamma_+$ be two continuous functions such that for all $t\geq0$, $-|\gamma_+(t)| \leq \gamma_-(t) \leq +|\gamma_+(t)|$ and $\tilde{\gamma}_+(t)\geq 0$.
Let $z$ be a solution of $\dot{z}(t) = -\gamma_+(t)z(t) + \gamma_-(t)$ satisfying $-1 < z_0 < +1$.
Then for all $t\geq0$ we have $-1 < z(t) < +1$. %$z$ atteint une valeur maximale (i.e. $\pm1$) en un instant $t>0$ si et seulement si $z$ est constante égale à $\pm1$.
%Dans ce cas, $\gamma_+ = \pm \gamma_-$.
\end{lemma}
\begin{proof}
Assume there exists a positive time for which $z$ is equal to 1. Let $t_0$ be the smallest of these times, $z(t_0) = 1$ (the case $z(t_0) = -1$ can be done along the same lines).
Since $-1 < z_0 < +1$,  we get $t_0 > 0$. Moreover, we necessarily have $\dot{z}(t)>0$ for any $t<t_0$ close enough to $t_0$, since the function $z$ increases toward $+1$ before reaching it.
We now show that $\dot{z}(t) \leq 0$ around $t_0$. Since $-|\gamma_+(t)| \leq \gamma_-(t) \leq +|\gamma_+(t)|$, we have
$$
-|\gamma_+(t)|(1 + z(t)) \leq \dot{z}(t) \leq +|\gamma_+(t)|(1 - z(t)).
$$
It follows that if there exists $t_0$ such that $z(t_0)=1$ then $\dot{z}(t_0)\leq0$ (Likewise, if there is some $t_0$ such that $z(t_0)=-1$ then $\dot{z}(t_0)\geq0$).
As we assume that $\gamma_\pm$ is continuous, the function $z(t)$ is continuously differentiable. Hence there is a neighborhood about $t_0$ such that $\dot{z}(t)\leq0$ (or $\dot{z}(t)\geq0$ if $z(t_0)=-1$). This is a contradiction. We conclude that $z(t) < +1$ for all $t\geq0$ (and likewise $z(t) > -1$).\qed
\end{proof}
\begin{lemma}
\label{lem_2}
Let $\gamma_-,\gamma_+$ be two continuous functions such that, for all $t\geq0$, $\tilde{\gamma}_+ \geq0\ \textrm{and}\ -|\gamma_+| \leq \gamma_- \leq +|\gamma_+|$.
Let $z$ be a solution of the differential equation $\dot{z}(t) = -\gamma_+(t)z(t) + \gamma_-(t)$ satisfying $-1\leq z_0 \leq +1$.
Then for all $t\geq0$ we have
$$
- (1 - (1 + z_0)e^{-\tilde{\gamma}_+(t)}) \leq z(t) \leq 1 - (1 - z_0)e^{-\tilde{\gamma}_+(t)}.
$$
In particular,
$$
- (1 - e^{-\tilde{\gamma}_+(t)}) \leq s(t) \leq 1 - e^{-\tilde{\gamma}_+(t)}.
$$
\end{lemma}
\begin{proof}
Suppose that $-1 < z(0) < +1$.
Then, according to Lemma~\ref{lem_1}, for all $t \geq 0$ we have $-1 < z(t) < +1$.

From the differential equation and the conditions satisfied by $\gamma_\pm$ we get that, for all $t \geq 0$,
$$
-|\gamma_+(t)|(1 + z(t)) \leq \dot{z}(t) \leq +|\gamma_+(t)|(1 - z(t)).
$$
We obtain the following inequalities $\frac{-\dot{z}}{1 - z} \geq -|\gamma_+|$ and $\frac{\dot{z}}{1 + z} \geq -|\gamma_+|$.
Integrating from $0$ to $t$, we get
$$
\ln\!\left(\frac{1\pm z}{1 \pm z_0}\right) \geq -\int_0^t|\gamma_+(u)|du = -|\tilde{\gamma}_+|.
$$
Hence, $1\pm z \geq (1\pm z_0)e^{-\tilde{\gamma}_+}$.
This leads to the result when $z_0\neq\pm1$.

Suppose now that $z_0 = \pm1$. Then $z$ can be written as $z(t) = s(t) \pm e^{-\tilde{\gamma}_+}$ with $s(0) = 0$.
From the above property, we arrive at
$- (1 - e^{-\tilde{\gamma}_+(t)}) \leq s(t) \leq 1 - e^{-\tilde{\gamma}_+(t)}$,
which gives the result when $z_0 = \pm1$.\qed
\end{proof}
\section{Asymptotic behavior of the dynamics}\label{secasympnew}
We prove the two propositions~\ref{prop4} and \ref{prop5} of Sec.~\ref{secnew}. We first consider a preliminary result.
\begin{lemma}
\label{lem:1}
Let $z$ be a bounded solution of $(E)$. If $\dot{z}$ converges to a finite limit then this limit is 0.
\end{lemma}
\begin{proof}
We denote by $\ell$ the limit of $\dot{z}$. Assume that $\ell>0$.
For all $\epsilon>0$ small enough, there exists an interval $[M;\,+\infty[$ such that $\dot{z}(t)\geq\epsilon$.
We consider a given value of $\epsilon$ (for instance $\epsilon = \ell/2$).
By integrating over the interval $[M;\,+\infty[$, we obtain
$$
z(t) - z(M) = \int_M^t\dot{z}(u)du \geq \epsilon(t-M).
$$
We deduce that $z$ goes to $+\infty$ and therefore is not bounded. So $\ell\leq0$. In the same way, we show that $\ell<0$ implies that $z$ converges to $-\infty$.
Finally, we get $\ell = 0$.\qed
\end{proof}
\noindent We can then show Prop.~\ref{prop4}.\\
\begin{proof}
We first consider the case $A_0>0$  and $B_0=0$.
For any $\epsilon>0$, there exists $M>0$ such that $\forall t\geq M$, $|A(t) - A_0|<\epsilon$ and $|B(t)|<\epsilon$.
We choose $\epsilon<A_0/2$. Assume that there exists $t_0\geq M$ such that $z(t_0) > \frac{\epsilon}{A_0 - \epsilon} >0$, then
$$
\dot{z}(t_0) = A(t_0)z(t_0) + B(t_0) > (A_0 - \epsilon)\frac{\epsilon}{A_0 - \epsilon} - \epsilon = 0.
$$
Let $[t_0,t]$ be the maximal closed interval on which $z$ is strictly increasing. We have $\dot{z}(t) = 0$. We deduce that $z(t) > z(t_0)$ and
$$
\dot{z}(t) = A(t)z(t) + B(t) > (A_0 - \epsilon)\frac{\epsilon}{A_0 - \epsilon} - \epsilon = 0,
$$
which leads to a contradiction unless $t = +\infty$. Since $z$ is strictly increasing on $[t_0,+\infty[$ and $z$ is bounded, we deduce that $z$ converges.

Assume that there exists $t_0\geq M$ such that $z(t_0) < -\frac{\epsilon}{A_0 + \epsilon} < 0$.
In the same way, we get that $z$ is a strictly bounded decreasing function on $[t_0,+\infty[$ and therefore a convergent function.

There remains the case for which for all $\epsilon>0$, there exists $M>0$ with $\forall t\geq M$, $|A(t) - A_0|<\epsilon$ and
$$
-\frac{\epsilon}{A_0 + \epsilon} \leq z(t) \leq \frac{\epsilon}{A_0 - \epsilon}
$$
We deduce that $z$ converges to 0. A similar argument can be used for $A_0<0$.

Now if $z$ converges, we know from (E) that $\dot{z}$ converges.
According to Lemma \ref{lem:1}, $\dot{z}$ converges to zero and so does $z$, since $A_0\neq0$.

Let us now assume that $B_0\neq 0$. We consider the change of variables $u(t) = z(t) - z_0(t)$ where $z_0$ is a bounded solution of the differential equation $\dot{z}_0 = A_0 z_0 + B_0$.
We get
\begin{eqnarray*}
\dot{u}(t) &=& \dot{z}(t) - \dot{z}_0(t) = A(t)z(t) + B(t) - A_0z_0(t) - B_0 \\
&=& A(t)(u(t) + z_0(t)) + B(t) - A_0z_0(t) - B_0\\
&=& A(t)u(t) + (A(t) - A_0)z_0(t) + B(t) - B_0\\
&=& A(t)u(t) + D(t)
\end{eqnarray*}
with $A(t) \to A_0$.
Note that $z_0 = Ce^{A_0t} -\frac{B_0}{A_0}$, where $C$ is an arbitrary constant.
Hence, if $A_0<0$ then $z_0$ is bounded and $D(t) \to 0$.
We then know from the preceding $B_0 = 0$ case that $u(t)$ converges to 0.
This means that $z$ converges to $-\frac{B_0}{A_0}$.\qed
\end{proof}
Note that if $A_0>0$ bounded solutions might not exist.
The proof above shows that when $A_0>0$ and when $A(t)$ converges to $A_0$ faster than $e^{-A_0t}$ then $D(t) \to 0$, hence $u(t) \to 0$, which means that the non-zero solutions $z$ diverge like $e^{A_0t}$ (they are very similar to the solutions of $\dot{z}_0 = A_0 z_0 + B_0$).

We consider now Prop.~\ref{prop5} and we first show a preliminary result.
\begin{lemma}
Let $z$ be a bounded solution of (E) with $A,B$ converging to 0. Then $z$ is a slowly oscillating function.
\end{lemma}
\begin{proof}
Let $\tau\in\mathbb{R}$ and $(x_n)$ a sequence going to $+\infty$. For all $n$, there exists $c_n \in [x_n,x_n+\tau]$ if $\tau>0$ (or $[x_n+\tau,x_n]$ if $\tau<0$) such that
$$
|z(x_n+\tau) - z(x_n)| = |\dot{z}(c_n)||\tau| = |A(c_n)z(c_n)+B(c_n)||\tau|
$$
which converges to 0 since $c_n\to\infty$ and $z$ is bounded.\qed
\end{proof}
\noindent We then prove Prop.~\ref{prop5}.\\
\begin{proof}
There exists $\alpha>0$ such that $A(t),B(t) \in o(\frac{1}{t^{1+\alpha}})$ (we can choose the same $\alpha$ for both functions by taking the smallest).
In other words, for all $\epsilon>0$, there exists $T>0$ such that, for all $t\geq T$, we have $|A(t)|<\frac{\epsilon}{t^{1+\alpha}}$ and $|B(t)|<\frac{\epsilon}{t^{1+\alpha}}$.

Let $\epsilon>0$ and $T>0$ as above. For any $x<t<y$ larger than $T$ we have
\begin{eqnarray*}
|z(y)-z(x)| &=& \left|\int_x^y \dot{z}(t)dt\right| \\
&\leq& \int_x^y \left| \dot{z}(t) \right|dt = \int_x^y \left| A(t)z(t) + B(t) \right|dt \\
&\leq& \int_x^y \left(|A(t)||z(t)| + |B(t)| \right)dt\\
&\leq& \epsilon \int_x^y \frac{1}{t^{1+\alpha}}(|z(t)| + 1)dt\\
&\leq& \epsilon M  \int_x^y \frac{1}{t^{1+\alpha}}dt \leq \epsilon M  \int_1^\infty \frac{1}{t^{1+\alpha}}dt\\
&\leq& \epsilon M/\alpha,
\end{eqnarray*}
where $M>|z(t)| + 1$. The function $z$ is bounded and it has at least one accumulation point. Let $\ell,\ell'$ two accumulation points, i.e. there exist two sequences $(x_n)$, $(y_n)$ going to $+\infty$ such that $z(x_n)\to\ell$ et $z(y_n)\to\ell'$. For $n$ large enough, $x_n$ and $y_n$ are larger than $T$ and thus, for all $\epsilon>0$, there exists $N\in\mathbb{N}$ such that, $n\geq N$ implies that  $|z(y_n)-z(x_n)|\leq \epsilon M/\alpha$. In the limit $n\to +\infty$, we obtain that for all $\epsilon>0$ we have $|\ell'-\ell|\leq \epsilon M/\alpha$, which leads to $\ell = \ell'$.
Since $z$ has a unique accumulation point, $z$ is a convergent function.\qed
\end{proof}

\end{document}